\documentclass[a4paper]{amsart}
\newtheorem{theorem}{Theorem}[]
\newtheorem{lemma}[theorem]{Lemma}

\newtheorem{proposition}[theorem]{Proposition}
\theoremstyle{definition}
\newtheorem{example}[theorem]{Example}
\newtheorem{definition}[theorem]{Definition}
\numberwithin{equation}{section}

\begin{document}

\title[A Note on The Enumeration of Euclidean Self-Dual Skew-Cyclic Codes]
 {A Note on The Enumeration of Euclidean Self-Dual Skew-Cyclic Codes over Finite Fields}

\author[Irwansyah]{Irwansyah}

\address{Mathematics Department\\
Faculty of Mathematics and Natural Sciences\\
Universitas Mataram, Mataram\\
INDONESIA}

\email{irw@unram.ac.id}


\author[I. M. Alamsyah, A. Muchlis, A. Barra]{Intan Muchtadi-Alamsyah, Ahmad Muchlis, Aleams Barra}

\address{Algebra Research Group\\
Faculty of Mathematics and Natural Sciences\\
Institut Teknologi Bandung, Jl. Ganesha 10, Bandung, 40132\\
INDONESIA}

\email{ntan,muchlis,barra@math.itb.ac.id}

\author[D. Suprijanto]{Djoko Suprijanto}
\address{Combinatorial Mathematics Research Group\\
Faculty of Mathematics and Natural Sciences\\
Institut Teknologi Bandung\\
Jl. Ganesha 10, Bandung, 40132\\
INDONESIA}
\email{djoko@math.itb.ac.id}
\subjclass{11T71}

\keywords{Skew-cyclic codes, Euclidean self-dual codes, Enumeration.}

\date{}
\dedicatory{}

\begin{abstract}
In this paper,
we give the enumeration formulas for Euclidean self-dual skew-cyclic codes over
finite fields when $(n,|\theta|)=1$ and for some cases when $(n,|\theta|)>1,$
where $n$ is the length of the code and $|\theta|$ is the order of automorphism $\theta.$ \\[0.25cm]
\end{abstract}

\maketitle
\section{Introduction}
Let $\mathbb{F}_q$ be the finite field of cardinality $q,$ where $q$ is a prime power.
A linear code over $\mathbb{F}_q$ with length $n$ is a linear subspace of
the $\mathbb{F}_q$-vector space $\mathbb{F}_q^n.$ The two classes of interesting linear
codes are cyclic codes and self-dual codes.
The reason is, among other things, because these classes of codes have produced many examples of optimal codes.
Therefore, some authors also studied the combination of these two classes of codes,
namely self-dual cyclic codes \cite{jia} and the generalization of cyclic codes
such as quasi-cyclic codes and skew-cyclic codes\cite{ling}.

Skew-cyclic codes or $\theta$-cyclic codes
over finite fields are a generalization of cyclic codes over finite fields.
Here, $\theta$ is an automorphism in the corresponding finite fields.
This class of codes has many interesting properties such as it can be viewed
as a left ideal or left submodule over a skew-polynomial ring, {\it etc}, see \cite{geis}.
Moreover, these codes also give $[38,18,11]$ code,
an Euclidean self-dual code over $\mathbb{F}_4,$
which improves the bound for optimal self-dual codes of length $36$ with respect to Hamming distance \cite{delphine}.

Recently, the study of special class of skew-codes over finite fields {\it i.e.}
Euclidean self-dual skew-codes are conducted regarding their existence and enumeration.
Boucher \cite{boucher15} showed that, when $q\equiv 1\mod 4,$ there always exists
a self-dual $\theta$-code in any dimension and that self-dual $\theta$-codes of
a given dimension are either all $\theta$-cyclic or all $\theta$-negacyclic.
Also, It is not exists when $q\equiv 3\mod 4.$ However,
the enumeration of this class of codes not yet been done completely.
So far, the enumeration was done for self-dual skew-codes over $\mathbb{F}_4$
with length $n=2^s$ \cite{delphine13} and for self-dual skew-codes
over $\mathbb{F}_{p^2}$ \cite{boucher14}.

In this paper, we study the enumeration of self-dual skew-cyclic or $\theta$-cyclic
codes when $(n,|\theta|)=1$ and for some cases when $(n,|\theta|)=d>1,$
where $|\theta|$ is the order of the automorphism $\theta.$

\section{The Enumeration Formulas}
Let us recall some definition which can be found in \cite{boucher14,geis,jia,ling}.
Let $T$ be a shift operator on $\mathbb{F}_q^n$ and $\theta$ is an automorphism in $\mathbb{F}_q.$

\begin{definition}
Let $C$ be a linear code of length $n$ over $\mathbb{F}_q.$
\begin{itemize}
\item[(1)]  If
\[T^l(c)=(c_{n-l},c_{n-l+1},\dots,c_{n-l-1})\in C\]
whenever $c=(c_0,c_1,\dots,c_{n-1})\in C,$ then $C$ is a quasi-cyclic code with index $l.$

\item[(2)] If
\[T_\theta(c)=(\theta(c_{n-1}),\theta(c_{0}),\dots,\theta(c_{n-2}))\in C\]
whenever $c=(c_0,c_1,\dots,c_{n-1})\in C,$ then $C$ is a $\theta$-cyclic code.

\item[(3)] If $\theta$ in (2) above is an identity map, then $C$ is a cyclic code.
\end{itemize}
\end{definition}

For any $c=(c_0,\dots,c_{n-1})$ and $c'=(c_0',\dots,c_{n-1}')$ in $\mathbb{F}_q,$
we define the Euclidean product between $c$ and $c'$ as
\[
[c,c'] = \sum_{i=0}^{n-1} c_ic_i'.
\]
Let $C$ be a code of length $n$ over $\mathbb{F}_q$ and $C^\bot=\{c'\in\mathbb{F}_q^n:~[c,c']=0,\forall c\in C\}.$
Then we can give the following definition.

\begin{definition}
A linear code $C$ is called \emph{Euclidean self-dual} if $C=C^\bot.$
\end{definition}

Now, before we do the enumeration, we need to proof the following proposition.

\begin{proposition}
If $C$ is a $\theta$-cyclic code over $\mathbb{F}_q$ of length $n,$
for some automorphism $\theta$ in $\mathbb{F}_q,$ then $C$ is either a cyclic code or a quasi-cyclic code.
\label{equiv}
\end{proposition}

\begin{proof}
We consider two cases for $|\theta|.$ First, if $\gcd(|\theta|,n)=1$, then there exist $p,q\in \mathbb{Z}$, such that
$pr+qn=1$ or $p|\theta|=1+bn$ for some $b$.
For any $c\in C$, let $c=(c_0,c_1,\dots,c_{n-1}),$ for some $c_0,c_1,\dots,c_{n-1}\in\mathbb{F}_q.$
Then we have
\[
\begin{aligned}
T_\theta^{p|\theta|}(c) & = \left(\theta^{p|\theta|}(c_{n-p|\theta|\;(\text{mod}\;n)}),
\theta^{p|\theta|}(c_{n-p|\theta|+1\;(\text{mod}\; n)}),\right.\\
&\quad \left.\dots,\theta^{p|\theta|}(c_{n-p|\theta|-1\;(\text{mod}\; n)})\right) \\
		    & = \left(\theta^{p|\theta|}(c_{n-1}),\theta^{p|\theta|}(c_{0}),\dots,\theta^{p|\theta|}(c_{n-2})\right) \\
            & = (c_{n-1},c_0,\dots,c_{n-2}) \\
            & = T(c),
\end{aligned}
\]
which means, $T(c)\in C$ or $C$ is a cyclic code.\\

Second, if $\gcd(|\theta|,n)=s,$ for some $s\in \mathbb{N},$ where $s\not=1$,
then there exist $p_1,p_2\in\mathbb{Z}$, such that $p_1|\theta|=s+p_2n$. Hence, we have
\[
\begin{aligned}
T_\theta^{p_1|\theta|}(c) & = \left(\theta^{p_1|\theta|}(c_{n-p_1|\theta|\;(\text{mod}\;n)}),
   \theta^{p_1|\theta|}(c_{n-p_1|\theta|+1\;(\text{mod}\; n)}),\right.\\
            & \quad \left. \dots,\theta^{p_1|\theta|}(c_{n-p_1|\theta|-1\;(\text{mod}\; n)})\right) \\
		    & = \left(\theta^{p_1|\theta|}(c_{n-s}),\theta^{p_1|\theta|}(c_{n-s+1}),\dots,\theta^{p_1|\theta|}(c_{n-s-1})\right) \\
            & = (c_{n-s},c_{n-s+1},\dots,c_{n-s-1})\\
            & = T^s(c)
\end{aligned}
\]
which means, $C$ is a quasi-cyclic code of index $s,$ as we hope.
\end{proof}

\subsection{Case 1 : $(n,|\theta|)=1$}

If $(n,|\theta|)=1$, then $\theta$-cyclic code is a cyclic code.
It means, the number of self-dual $\theta$-cyclic codes of length $n$ is
less than or equal to the number of self-dual cyclic codes of length $n.$
Recall that, based on the result in \cite{jia},
Euclidean self-dual cyclic code of length $n$ over $\mathbb{F}_q$ exist only
when $q$ is a power of 2 and $n$ is an even integer.
Now, let $n=2^{v(n)}\tilde{n},$ where $(\tilde{n},2)=1,$ and
let $\phi$ be the Euler function and $ord_j(i)$ be the smallest
integer $e$ such that $j$ divides $i^e-1.$ Also, let $j$ be an odd
positive integer and $m$ be a positive integer.
We say that the pair $(j,m)$ is {\it good} if $j$ divides $(2^m)^k+1$
for some integer $k\geq0$ and {\it bad} otherwise.
We define the function $\chi$ as follows.

\[
\chi(j,m)=
\begin{cases}
   0, & \text{if}\;(j,m)\;\text{good},\\
   1, & \text{otherwise}.
\end{cases}
\]
Recall that, based on \cite{jia}, the number of Euclidean self-dual cyclic codes over $\mathbb{F}_q$ is
\begin{equation}
\left(1+2^{v(n)}\right)^{\frac{1}{2}\sum_{j|\tilde{n}}\chi(j,m)\phi(j)/ord_j(2^m)}.
\label{numb}
\end{equation}
Note that, by Proposition~\ref{equiv}, when $(n,|\theta|)=1,$
then the number of Euclidean self-dual $\theta$-cyclic codes over $\mathbb{F}_q$
is less than or equal to the one provided by~(\ref{numb}).

Now, we make partitions of the set $\{0,1,\dots,\tilde{n}\}$
into $2^m$-cyclotomic cosets denoted by $C_s,$ where $s$ is the smallest element in $C_s.$
Let $\mathcal{A}$ be the collection of union of $2^m$-cyclotomic cosets
which represent the generators of the self-dual cyclic codes over $\mathbb{F}_q$
of length $n.$ Also, let $\theta(\beta)=\beta^{2^r},$ for all
$\beta\in\mathbb{F}_{2^m}$ and for some $r\in \{1,2,\dots,m\}.$
Also, let $\lambda$ be the map as follows
\[
\begin{array}{lcll}
\lambda_r: & \{0,1,\dots,\tilde{n}\} & \longrightarrow & \{0,1,\dots,\tilde{n}\} \\
 & a & \longmapsto & 2^ra \pmod{\tilde{n}}.
\end{array}
\]
Furthermore, we let $\Lambda_r$ be the map on $\mathcal{A}$
induced by $\lambda_r$ and $\overline{\Lambda}_r$ be the number of elements $A$
in $\mathcal{A}$ such that $\Lambda_r(A)\not= A.$ Then, we have the following result.

\begin{theorem}
If $(n,|\theta|)=1,$ $q=2^m,$ for some $m,$ and $n=2^{v(n)}\tilde{n},$
for some odd positive integer $\tilde{n},$ then the number of
Euclidean self-dual $\theta$-cyclic codes of length $n$ over $\mathbb{F}_q$ is
\[
\left(1+2^{v(n)}\right)^{\frac{1}{2}\sum_{j|\tilde{n}}\chi(j,m)\phi(j)/ord_j(2^m)}-\overline{\Lambda}_r.
\]
\label{numbcyclic}
\end{theorem}

\begin{proof}
Let $C=\langle g(x)\rangle,$ where $g(x)$ is a monic polynomial with minimum degree in $C$ and,
therefore, $g(x)$ is a divisor of $x^{\tilde{n}}-1.$
Now, for any $f(x)=\sum_{i=0}^sf_ix^i\in \mathbb{F}_{2^m}[x],$ let
\[
\theta(f)=\sum_{i=0}^s\theta(f_i)x^i=\sum_{i=0}^sf_i^{2^r}x^i.
\]
Also, let $\theta(C)=\langle \theta(g)\rangle.$ It is easy to see that $\theta(C)$ is also
a Euclidean self-dual cyclic code. Now, let $\alpha$ be a primitive
$\tilde{n}$-th root of unity in some extension of $\mathbb{F}_{2^m}.$
Then, by \cite[Theorem 4.2.1(vii)]{huffman}, we have
\[
g(x)=\prod_s M_{\alpha^s}(x),
\]
where $M_{\alpha^s}(x)$ is the minimal polynomial of $\alpha^s$ over $\mathbb{F}_{2^m},$
and $s$ is the representatives of the $2^m$-cyclotomic
cosets modulo $\tilde{n}.$ Now, let
\[
x^{\tilde{n}}-1=f_1(x)\cdots f_s(x)h_1(x)h_1^*(x)\cdots h_t(x)h_t^*(x),
\]
where $f_i(x)$ $(1\leq i\leq s)$ are monic irreducible self-reciprocal polynomials
over $\mathbb{F}_{2^m},$ while $h_j(x)$ and its reciprocal polynomial $h_j^*(x)$ $(1\leq j\leq t)$
are both monic irreducible polynomials over $\mathbb{F}_{2^m}.$ Then, by \cite[Theorem 3.7.6]{huffman}, we have
\[
f_i(x)=\prod_{k\in C_{s_i}}(x-\alpha^k),
\]
and
\[
h_j(x)=\prod_{k\in C_{s_j}}(x-\alpha^k).
\]
Therefore, in order to make $\theta(C)\subseteq C,$
we have to choose $f_i(x)$'s and $h_j(x)$'s such that their
representatives in $\mathcal{A}$ are fixed by $\Lambda_r,$
because $\theta(x-\alpha^l)=x-\alpha^{2^rl}.$ Since $\overline{\Lambda}_r$
is the number of non-fixed-by-$\Lambda_r$-elements in $\mathcal{A},$
we have the desired formula.
\end{proof}

Let us consider the following examples.

\begin{example}
Let $q=4, r=1,$ and $n=6.$ Therefore, $\tilde{n}=3.$
Also, let $\mathbb{F}_4=\mathbb{F}_2[\alpha],$ where $\alpha$ is the root
of the polynomial $x^2+x+1.$ Note that, $\alpha$ is also a primitive $3^{\text{th}}$-root of unity.
As we can check, we have three $4$-cyclotomic cosets modulo $3,$ {\it i.e.}
$C_0=\{0\},C_1=\{1\},$ and $C_2=\{2\}.$ Moreover, $C_0$
represents a self-reciprocal polynomial, while $C_1$ and $C_2$
represent two polynomials which reciprocal to each other.
Since $n=2\tilde{n},$ we have
\[
\mathcal{A}=\{C_0\cup C_1\cup C_1, C_0\cup C_2\cup C_2,C_0\cup C_1\cup C_2\}.
\]
We can see that $\Lambda_1(C_0)=C_0,\Lambda_1(C_1)=C_2,$ and $\Lambda_1(C_2)=C_1.$
Therefore, we have $\overline{\Lambda}_1=2,$ and the number of
self-dual $\theta$-cyclic codes is $|\mathcal{A}|-\overline{\Lambda}_1=1.$
\end{example}

\begin{example}
Let $q=4, r=1,$ and $n=14.$ Therefore, $\tilde{n}=7.$
Let $\alpha\in\mathbb{F}_{4^2}$ be a primitive $7^{\text{th}}$-root of unity.
We can check that we have three $4$-cyclotomic cosets modulo $7.$
They are $C_0=\{0\},C_1=\{1,2,4\},$ and $C_3=\{3,5,6\}.$
Also, $C_0$ represents a self-reciprocal polynomial, while $C_1$ and $C_3$
represent two polynomials which reciprocal to each other. So, we have that
\[
\mathcal{A}=\{C_0\cup C_1\cup C_1,C_0\cup C_3\cup C_3,C_0\cup C_1\cup C_3\}.
\]
We also have $\overline{\Lambda}_1=0,$ because $\Lambda_1(C_i)=C_i$ for all $i=0,1,3.$
Therefore, the number of self-dual $\theta$-cyclic codes is $|\mathcal{A}|=3.$
\end{example}

We have to note that $|\mathcal{A}|$ is the number of self-dual cyclic
codes of length $n$ over $\mathbb{F}_q.$
Therefore, the two examples above give us another
way to count the number of self-dual cyclic codes using
$q$-cyclotomic cosets of modulo $\tilde{n}.$

\subsection{Case 2 : $(n,|\theta|)=d>1$}

If $(n,|\theta|)=d>1$, then $\theta$-cyclic code is
a quasi-cyclic code of index $d$ as stated in Proposition~\ref{equiv}.
Now, for any $C\subseteq \mathbb{F}_q^n,$ define $T_\theta(C)=\{T_\theta(c):~\forall c\in C\}.$
Then, we can easily prove the following lemma.

\begin{lemma}
If $C$ is a Euclidean self-dual quasi-cyclic code of index $d,$
then $T_\theta(C)$ is also an Euclidean self-dual quasi-cyclic code of index $d.$
Moreover, $T_\theta(C)$ has the same Hamming weight enumerator polynomial as $C.$
\label{same}
\end{lemma}

The Lemma~\ref{same} above shows that the map $T_\theta$ preserves
Euclidean self-duality and quasi-cyclic property.
Moreover, we have to note that, if $C$ is a quasi-cyclic code such that
$T_\theta(C)\subseteq C,$ then $C$ is also a $\theta$-cyclic code.

For any code $C\subseteq \mathbb{F}_q^n,$ let $\rho$ be the map as follows,
\begin{equation}
\rho(C)=
\begin{cases}
0, & \text{if} T_\theta(C)\subseteq C,\\
1, & \text{otherwise}.
\end{cases}
\label{charmap}
\end{equation}

Let us recall structures of quasi-cyclic code described in \cite{ling}.
Let $n=dm,$ $R=\mathbb{F}_q[Y]/(Y^m-1),$ and $C$ be a quasi-cyclic code
over $\mathbb{F}_q$ of length $dm$ with index $d.$ Let
\[
\mathbf{c}=(c_{00},c_{01},\dots,c_{0,l-1},c_{10},\dots,c_{1,l-1},\dots,c_{m-1,0},\dots,c_{m-1,l-1})\in C.
\]
Then, the map $\phi : \mathbb{F}_q^{dm}\longrightarrow R^d$ defined by
\[
\phi(\mathbf{c})=(\mathbf{c}_0(Y),\mathbf{c}_1(Y),\dots,\mathbf{c}_{l-1}(Y))\in R^d,
\]
where $\mathbf{c}_j(Y)=\sum_{i=0}^{m-1}c_{ij}Y^i\in R,$ is
a one-to-one correspondence between quasi-cyclic codes over $\mathbb{F}_q$
of index $d$ and length $dm$ and linear codes over $R$ of length $d.$ Moreover,
\[
R^d = \left(\bigoplus_{i=1}^sG_i^d\right)\oplus\left(\bigoplus_{j=1}^t\left(H_j'^d\oplus H_j''^d\right)\right),
\]
where $G_i=\mathbb{F}_q[Y]/(g_i), H_j'=\mathbb{F}_q[Y]/(h_j),$ and $H_j''=\mathbb{F}_q[Y]/(h_j^*),$
for some self-reciprocal irreducible polynomials $g_i$ $(1\leq i\leq s)$
and irreducible reciprocal pairs $h_j$ and $h_j^*$ $(1\leq j\leq t)$ which satisfy
\[
Y^m-1=\delta g_1\cdots g_s h_1h_1^*\cdots h_th_t^*.
\]
The above decomposition gives
\[
C=\left(\bigoplus_{i=1}^sC_i\right)\oplus\left(\bigoplus_{j=1}^t\left(C_j'\oplus C_j''\right)\right),
\]
where $C_i$ is a linear code over $G_i$ of length $d,$ $C_j'$ is a linear code over $H_j'$
of length $d,$ and $C_j''$ is a linear code over $H_j''$ of length $d.$

Now, let
\[
\rho_{G_i}=\sum_{C\;\text{code over}\;G_i}\rho(\phi^{-1}(C)),
\]
\[
\rho_{H_j',H_j''}=\sum_{{C' \;\text{code over}\;H_j'\atop C''\;\text{code over}\;H_j''}}\rho(\phi^{-1}(C'\oplus C'')),
\]
and
\[
N(d,q)=1+\displaystyle{\sum_{k=1}^d\frac{(q^d-1)(q^d-q)\cdots(q^d-q^{k-1})}{(q^k-1)(q^k-q)\cdots(q^k-q^{k-1})}}.
\]
Then, we have the following result for $d=2.$

\begin{proposition}
Let $\theta$ be an automorphism in $\mathbb{F}_q$ such that $(n,|\theta|)=2,$
where $q$ be a prime power satisfying one of the following conditions,
\begin{enumerate}
\item $q$ is a power of 2,
\item $q=p^b,$ where $p$ is a prime congruent to $1\;\text{mod}\;4,$ or
\item $q=p^2b,$ where $p$ is a prime congruent to $3\;\text{mod}\;4.$
\end{enumerate}
Also, let $m$ be an integer relatively prime to $q.$
Suppose that $Y^m-1=\delta g_1\cdots g_s h_1h_1^*\cdots h_th_t^*$ as mentioned above.
Suppose further that $g_1=Y-1$ and, if $m$ is even, $g_2=Y+1.$ Let the degree
of $g_i$ be $2d_i$ and the degree of $h_j$ (also $h_j^*$) be $e_j.$
Then, the number of distinct Euclidean self-dual $\theta$-cyclic codes of
length $2m$ with index 2 over $\mathbb{F}_q$ is

\[
\begin{array}{ll}
\displaystyle{4\prod_{i=3}^s(q^{d_i}+1-\rho_{G_i})\prod_{j=1}^t(N(2,q^{e_j})-\rho_{H_j',H_j''})},
& \text{if}\;m\;\text{is even}\;\text{and}\;q\;\text{is odd}\\

\displaystyle{2\prod_{i=2}^s(q^{d_i}+1-\rho_{G_i})\prod_{j=1}^t(N(2,q^{e_j})-\rho_{H_j',H_j''})},
& \text{if}\;m\;\text{is odd}\;\text{and}\;q\;\text{is odd}\\

\displaystyle{\prod_{i=2}^s(q^{d_i}+1-\rho_{G_i})\prod_{j=1}^t(N(2,q^{e_j})-\rho_{H_j',H_j''})},
& \text{if}\;m\;\text{is odd}\;\text{and}\;q\;\text{is even}\\
\end{array}
\]
\end{proposition}

\begin{proof}
Apply \cite[Proposition 6.2]{ling} and use the fact that the map $\rho$
counts the number of codes which are not invariant under the action of $T_\theta.$
\end{proof}

Furthermore, using \cite[Proposition 6.6, Proposition 6.9, Proposition 6.10, Proposition 6.12, and Proposition 6.13]{ling}
and the fact about the map $\rho$ as above, we have the following results.

\begin{proposition}
Suppose $q\equiv 1\;\text{mod}\;4$ and $d$ is even, or $q\equiv 3\;\text{mod}\;4$
and $d \equiv 0\;\text{mod}\;4.$ Then the number of Euclidean self-dual
$\theta$-cyclic codes of length $2d$ over $\mathbb{F}_q$ is
\[
\displaystyle{4\prod_{i=1}^{\frac{d}{2}-1}(q^i+1-\rho_{G_i})^2}.
\]
\end{proposition}

\begin{proposition}
Suppose that $q$ and $d$ satisfy one of the following conditions,
\begin{enumerate}
\item $q\equiv 11\;\text{mod}\;12$ and $d\equiv 0\;\text{mod}\;4,$
\item $q\equiv 2\;\text{mod}\;3$ but $q\not\equiv 11\;\text{mod}\;12$ and $d$ is even.
\end{enumerate}
Then the number of distinct Euclidean self-dual $\theta$-cyclic codes over $\mathbb{F}_q$ of length $3d$ is

\[
\displaystyle{b(q+1)\prod_{i=1}^{\frac{d}{2}-1} (q^i+1-\rho_{G_i})(q^{2i+1}+1-\rho_{H_i'})},
\]
where $b=1$ if $q$ is even, 2 if $q$ is odd.
\end{proposition}

\begin{proposition}
Let $q$ and $d$ satisfy one of the following conditions,
\begin{enumerate}
\item $q\equiv 7\;\text{mod}\;12$ and $d\equiv 0\;\text{mod}\;4,$
\item $q\equiv 1\;\text{mod}\;3$ but $q\not\equiv 7\;\text{mod}\;12$ and $d$ is even.
\end{enumerate}
Then the number of distinct Euclidean self-dual $\theta$-cyclic codes over $\mathbb{F}_q$ of length $3d$ is

\[
\displaystyle{b\left(\prod_{i=1}^{\frac{d}{2}-1} (q^i+1-\rho_{G_i})\right)(N(d,q)-\rho_{H',H''})},
\]
where $b=1$ if $q$ is even, 2 if $q$ is odd.
\end{proposition}

\begin{proposition}
Let $q$ be an odd prime power such that $-1$ is not a square in $\mathbb{F}_q$
and let $d\equiv 0\;\text{mod}\;4.$ Then the number of distinct Euclidean
self-dual $\theta$-cyclic codes over $\mathbb{F}_q$ of length $4d$ is
\[
\displaystyle{4(q+1)\prod_{i=1}^{\frac{d}{2}-1} (q^i+1-\rho_{G_i})^2(q^{2i+1}+1-\rho_{H_i'})}.
\]
\end{proposition}

\begin{proposition}
Let $d$ be an even integer and $q$ be an odd prime power such that $-1$
is a square in $\mathbb{F}_q.$ Then, the number of distinct Euclidean
self-dual $\theta$-cyclic codes of length $4d$ over $\mathbb{F}_q$ is
\[
\displaystyle{4\left(\prod_{i=1}^{\frac{d}{2}-1} (q^i+1-\rho_{G_i})^2\right)(N(d,q)-\rho_{H',H''})}.
\]
\end{proposition}

\section{Conclusion}
Enumeration of skew-cyclic or $\theta$-cyclic codes over finite fields has been considered by
Boucher and her coauthors in \cite{delphine13}, \cite{boucher14}.  However, the enumeration of this class
of codes has not yet been done completely. In this paper,
we study the enumeration of self-dual skew-cyclic or $\theta$-cyclic
codes if $(n,|\theta|)=1$ and for some cases if $(n,|\theta|)=d>1,$
where $|\theta|$ is the order of the automorphism $\theta.$

\section{Acknowledgment}
I, I.M-A, and A.B. are supported in part by Riset Unggulan Perguruan Tinggi Dikti 2016.
D.S. is supported in part by Riset ITB 2016.

\end{document}